\documentclass[12pt]{article}
\usepackage[utf8]{inputenc}
\usepackage{natbib}
\setcitestyle{authoryear,open={(},close={)}}

\setcitestyle{authoryear}
\def\cite{\citet}

\usepackage{comment}
\usepackage{xcolor}
\usepackage{tikz}
\usepackage{color}
\usepackage{amsmath}
\usepackage{amsthm}
\usepackage{amssymb}
\usepackage{amsfonts}
\usepackage{geometry}
\usepackage[normalem]{ulem}
\usepackage{float}

\usepackage{hyperref}
\hypersetup{colorlinks=true,linkcolor=blue,citecolor=blue}
\usepackage{dsfont}

\usepackage{nicefrac}

\def\Re{\mathbf{R}}

\newcommand{\df}[1]{{\em #1 }}

\definecolor{ForestGreen}{rgb}{.13,.54,.13}
\definecolor{BrickRed}{rgb}{.80,.26,.33}

\newtheorem{theorem}{Theorem}
\newtheorem{lemma}{Lemma}

\newtheorem{corollary}{Corollary}
\newtheorem{proposition}{Proposition}
\newtheorem{criterion}{Criterion}

\newtheorem{definition}{Definition}

\newcommand{\R}{\mathbb{R}}
\newcommand{\Z}{\mathbb{Z}}

\renewcommand{\Re}{\mathbb{R}}

\newcommand{\supp}{\mathrm{supp}}

\theoremstyle{definition}
\newtheorem{example}{Example}

\theoremstyle{remark}
\newtheorem{remark}{Remark}

\title{Efficiency in Random Resource Allocation and Social Choice\footnote{We are grateful to (in
alphabetic order) Ron Holzman, Luciano Pomatto, and  Omer Tamuz for discussions that inspired this work and {to Haris Aziz, Felix Brandt, and anonymous reviewers for pointing out important literature omissions in the earlier draft. We thank 
Samson Alva, Anna Bogomolnaia, Herv\'{e} Moulin,
and seminar audiences at the Universidad de Chile and the BFI Conference in Honor of Hugo Sonnenschein for useful comments, and the Linde Institute at Caltech for its financial support.}}}

\author{ \large Federico Echenique\thanks{UC Berkeley. Email:  fede@econ.berkeley.edu. Federico Echenique thanks the National Science Foundation for its support through the grants SES~1558757 and  CNS 1518941.} \ \ \ \ Joseph Root\thanks{University of Chicago. Email:  jroot@uchicago.edu. Joseph Root was partially supported by PIMCO.} \ \ \ \ Fedor Sandomirskiy\thanks{Caltech. Email:  fsandomi@caltech.edu.
Fedor Sandomirskiy was partially supported by the National Science Foundation (grant  CNS 1518941).}}
\date{}

\begin{document}

\maketitle

\begin{abstract}

We study efficiency in general collective choice problems where agents have ordinal preferences and randomization is allowed.  
{We explore the structure of preference profiles where ex-ante and  ex-post efficiency coincide, offer a unifying perspective on the known results, and give several new characterizations.} The results have implications for well-studied mechanisms including random serial dictatorship and a number of specific environments, including the dichotomous, single-peaked, and social choice domains. 
\end{abstract}

\section{Introduction}

Efficiency is among the central desiderata of mechanism and market design. There are at least two reasons for this. First, mechanisms are designed to maximize the welfare of the agents it is meant to serve. When it is possible to make all agents better-off, it is hard to defend not doing so. However, there is generally a trade-off between the simplicity, fairness and efficiency of mechanisms. When these desiderata clash, the designer is forced to compromise. For example, in school choice the most widely used mechanism, deferred acceptance, is fair and simple to use, but not efficient. A second reason for designing efficient mechanisms is that efficiency grants a kind of stability to the outcome. If it were possible for the agents involved (e.g. schools and students, doctors and hospitals, buyers and sellers) to construct an alternative mechanism which could make them all better-off, one might expect them to do so. The mechanism itself is a means by which agents coordinate to improve on a market imperfection, so if one mechanism leaves welfare available to the agents why wouldn't they simply design a second mechanism?

To achieve fairness in discrete settings, it is common to use randomization. A central planner uses a randomization device to decide the outcome. For example, they may flip a coin to decide which of two agents gets a desirable resource, or they might randomly break ties in priorities in deferred acceptance. With randomization and ordinal preferences there are (at least) two natural notions of efficiency: ``ex-post efficiency'' and ``ex-ante efficiency\footnote{Ex-ante efficiency is sometimes called ``ordinal efficiency.''}.'' Ex-post efficiency simply specifies that the randomization device place positive probability only on outcomes that are Pareto optimal as deterministic outcomes. Thus, regardless of the realization of the randomization device, there will be no way to make all agents better-off ex-post. Ex-ante efficiency is stronger. It requires that the random allocation itself cannot be improved; there should be no alternative randomization device which gives all agents a better lottery over outcomes.

The difference between the two notions may not be apparent at first. The following example makes the point in the context of Arrovian social choice. 

\begin{example}\label{ex_2Condorcet}
Consider three agents and six outcomes $X=\{a,b,c,x,y,z\}$. Agents' preferences are as follows
\begin{equation}\label{eq_2Condorcet}
\begin{array}{c}
   a\succ_1  x \succ_1  b \succ_1  y \succ_1  c \succ_1  z\\
   b \succ_2  z \succ_2  c \succ_2  x \succ_2  a \succ_2  y\\
   c \succ_3 y \succ_3 a \succ_3 z \succ_3 b \succ_3 x
\end{array}.
\end{equation}

In words, there are two interlacing Condorcet cycles: one for the outcomes $a,b$ and $c$ and one for the outcomes $x,y$ and $z$. For each pair of outcomes, there is an agent preferring one outcome to the other. As a consequence, each outcome is Pareto optimal, and therefore all lotteries over outcomes are ex-post efficient. We claim that the lottery $p'=\frac{1}{3}x+\frac{1}{3}y+\frac{1}{3}z$ is not ex-ante efficient. Consider the alternative lottery $p=\frac{1}{3}a+\frac{1}{3}b+\frac{1}{3}c$. In $p$, each agent gets a uniform random draw from her 1st, 3rd, and 5th ranked outcomes. By contrast, under $p'$, they get a uniform random draw between her 2nd, 4th, and 6th ranked outcomes. Clearly $p$ improves welfare relative to $p'$. Formally, $p$ is ex-ante efficient and $p'$ is not. 
\end{example}

Despite the fact that ex-post efficiency can leave welfare on the table, situations like those of Example~\ref{ex_2Condorcet} emerge naturally. Consider two ex-ante efficient lotteries $p$ and $p'$ with some agents preferring $p$ and others preferring $p'$. A natural compromise is to implement $\frac{1}{2}p+\frac{1}{2}p'$, giving both sides an equal chance of getting their preferred outcome. This compromise lottery will be ex-post efficient. The problem is that it may no longer be ex-ante efficient.

\begin{example}\label{ex_dichotomous}
Five agents must choose one of four outcomes $a,b,c$ and $d$. Each agent either approves or disapproves of each outcome. Preferences are listed in the following array, with approval being denoted by 1.
\[ 
\begin{array}{c|cccc}
 & a & b & c & d\\
 \hline
1 & 1 & 0 & 1 & 0 \\
2 & 0 & 1 & 0 & 1 \\
3 & 1 & 0 & 0 & 1  \\
4 & 1 & 0 & 0 & 0  \\
5 & 0 & 1 & 1 & 0 \\
\end{array}
\]
Note that all four outcomes are Pareto optimal. In particular, the degenerate lotteries $c$ and $d$ are ex-ante efficient. However, $\frac{1}{2}c+\frac{1}{2}d$ is dominated by $\frac{1}{2}a+\frac{1}{2}b$. 

\end{example}

This is why, despite its widespread use, random serial dictatorship (RSD)\footnote{In random serial dictatorship, agents are chosen one-at-a-time uniformly at random to select their favorite outcome(s) consistent with the choices of earlier dictators. See Section~\ref{sec_model} for a precise description.} is, in general, not ex-ante efficient. For any realization of the dictatorship ordering, the outcome is Pareto optimal, but from an ex-ante point of view, all agents may be made better off. Consider RSD for the preferences from Example~\ref{ex_dichotomous}. If agent 2 is selected as the first dictator, she will limit the outcome to either $b$ or $d$, the two outcomes she approves. If agent $3$ is chosen as the second dictator, since he prefers $d$ to $b$, the outcome of RSD will be $d$\footnote{Later dictators' preferences don't matter since a single outcome has been selected}. Likewise, if the first dictator is $1$ and the second is $5$, the outcome of RSD will be $c$. Hence RSD will place strictly positive probability on both $d$ and $c$. That is, there is some $\alpha>0$ such that both $c$ and $d$ are chosen with probability at least $\alpha$. Then, from Example~\ref{ex_dichotomous}, RSD cannot be ex-ante efficient, since replacing $\alpha c + \alpha d$ with $\alpha a + \alpha b$ will improve welfare. 

In this paper, we describe the settings where compromises of the form $\frac{1}{2}p+\frac{1}{2}p'$ are ex-ante efficient if $p$ and $p'$ are. That is, settings where the set of ex-ante efficient lotteries is convex {which happens  if and only if every ex-post efficient lottery is ex-ante efficient.} For these settings, RSD is immediately seen to be ex-ante efficient. In environments where the set of ex-ante efficient lotteries is not convex, there are still preference profiles for which ex-ante and ex-post efficiency coincide. For example, if all agents top-rank a given outcome, the degenerate lottery on that outcome is the only ex-ante efficient lottery. 
{We describe  those preference profiles and preference domains in which ex-post and ex-ante efficiency coincide. Our approach offers a unifying perspective on known results and implies some new characterizations.}

\subsection{Related Literature}

We study the problem of choosing a random outcome on the basis of agents' ordinal preferences. To our knowledge, the first paper to consider this problem is \citet{gibbard1977manipulation} which studies Arrovian social choice. \citet{gibbard1977manipulation} is an attempt to circumvent the impossibilities of the Gibbard-Satterthwaite theorem (\cite{gibbard1973manipulation}, \cite{satterthwaite1975strategy}) by allowing the mechanism to choose a random outcome. \citet{fishburn1984probabilistic} studies the same environment and showed that allowing for a randomized outcome restores the existence of a Condorcet winner. Probabilistic social choice has also attracted renewed interest in recent years. \citet{brandt2017rolling} provides an overview.

\citet{bogomolnaia2001new} initiates the study of ordinal mechanisms in the context of house allocation (sometimes called ``one-sided matching''). They highlight the difference between ex-post and ex-ante efficiency, showing that no mechanism can simultaneously be fair, ex-ante efficient and strategy-proof. {They formulate a conjecture playing a crucial role in our analysis:} ex-ante efficiency is equivalent to the maximization of social welfare for some utilitarian representation of the ordinal preferences.  \cite{mclennan2002ordinal} confirms the conjecture using a variant of Farkas' Lemma and provides an elegant geometric picture of ex-ante efficiency.
\cite{manea2008constructive} finds a simpler proof of this result, using a characterization of ex-ante efficiency by way of acyclicity due to \cite{KATTA2006231}. \citet{abdulkadirouglu2003ordinal} give an alternative combinatorial description of ex-ante efficiency in terms of ``dominated sets.'' {\cite{carroll2010efficiency} demonstrates the equivalence of ex-ante efficiency and welfare maximization for general domain with indifferences and \cite{aziz2015universal} extend the result beyond von Neumann-Morgenstern utilities.}

We seek to describe the particular preference profiles under which {the sets of ex-ante and ex-post efficient lotteries coincide. The literature has mainly focused on a related question of checking that a given lottery is ex-ante efficient. \cite{aziz2015universal} explores  structural properties of ex-ante efficient lotteries, \cite{aziz2014characterization} derives a combinatorial characterization of ex-ante efficient lotteries which is refined in our paper, \cite{aziz2018tradeoff} study compatibility of ex-ante efficiency with non-manipulability in various domains, and
\cite{maneaTE2009} investigates the fraction of profiles in the housing model where ex-ante efficiency is equivalent to ex-post and shows that this fraction vanishes as the number of agents and objects grows.}

In the dichotomous domain, the difference between ex-ante and ex-post efficiency was first studied by \cite{bogomolnaia2005collective}, who show that RSD is not ex-ante efficient (the focus in their paper is on the simultaneous satisfaction of fairness, incentive, and efficiency properties). 
\cite{itaybendan} finds a condition on a profile of dichotomous preferences for ex-ante and ex-post efficiency to coincide, formulated in terms of balanced coalitions. Item~(3) in our Proposition~\ref{prop:characterizdichotomous} is similar to his result. {Ex-ante efficiency in the dichotomous domain is  studied further by \cite{duddy2015fair} and \cite{aziz2019fair}.}

\citet{zhou1990conjecture} was the first to point out that RSD is not  ex-ante efficient when agents have von Neumann-Morgenstern expected utility preferences. \citet{bogomolnaia2001new} showed that RSD isn't ex-ante efficient with ordinal preferences.

\section{Preliminaries}\label{sec_model}
We describe a general model of collective decision making, allowing agents to have arbitrary (rational) preferences over outcomes. Preferences may display indifference between different outcomes, so the model has  the standard private-good ``housing model,'' as well as the public-good ``social choice model,'' as special cases. We define concepts of ex-ante and ex-post efficiency, and introduce basic tools to analyze them which will be applied in the subsequent sections to particular domains of preferences. 

Let $X$ be a finite set of \df{outcomes}. A \df{preference} over $X$ is a weak order; that is, a complete and transitive binary relation over $X$. When $\succeq$ is a preference, an \df{upper contour set} of $\succeq$ is any set of the form $\{x'\in X:x'\succeq x\}$, for an outcome $x\in X$. A \df{utilitarian representation} of a preference $\succeq$ is a function $u:X\to\Re$ so that $x\succeq y$ iff $u(x)\geq u(y)$. A preference $\succeq$ is said to be \df{strict} if $x\succeq y$ and $y\succeq x$ imply that $x=y$. A preference is \df{dichotomous} if it has a utilitarian representation $u$ with $u(X)\subseteq \{0,1\}$. When $\succeq$ is a preference over $X$, we denote the strict part of $\succeq$ by $\succ$, and the indifference relation by $\simeq$. 
 
Let $N$ then be a finite population of agents, each agent $i$ endowed with a preference relation $\succeq_i$ over $X$. The list of preferences $(\succeq_i)_{i\in N}$ is a \df{preference profile.}

A \df{lottery} over  $X$ is a probability distribution over $X$. The set $\Delta(X)=\{p\in\Re^X_+: \sum_{x\in X}p_x =1\}$ collects all lotteries over $X$.

An outcome $x\in X$ is \df{Pareto optimal} if it cannot be improved upon for any of the agents without harming other agents, i.e., there is no other outcome $y\in X$ such that $y\succeq_i x$ for all agents $i\in N$ and $y\succ_j x$ for some $j\in N$. The set of all Pareto-optimal outcomes is denoted by $X^*$. 
 
A simple way to obtain a Pareto-optimal outcome is via Serial Dictatorship. Given an ordering of the agents, the first agent selects the set of her most preferred outcomes, the second agent refines this choice by selecting her most preferred outcomes among those chosen by the first agent, and so on. Formally, consider a permutation $\sigma$ of the set of agents $N$. The lexicographic preference $\succeq_\mathrm{lex}$ associated with a preference profile $(\succeq_i)_{i\in N}$ and ordering of agents $\sigma$ is defined as follows: $x\succ_\mathrm{lex} y$ if there exist $k\in \{1,\ldots, n\}$ such that for $x\simeq_{\sigma(l)} y$ for $l\in \{1,\ldots, k-1\}$  
 and $x\succ_{\sigma(k)} y$.
 The serial dictatorship for a preference profile $(\succeq_i)_{i\in N}$ and ordering of agents $\sigma$ outputs all the the maximal elements of the corresponding lexicographic preference $\succeq_\mathrm{lex}$:
 $$\mathrm{SD}\big((\succeq_i)_{i\in N},\sigma\big)=\arg\max_X\, \succeq_\mathrm{lex}.$$
 It is well-known that any $x\in \mathrm{SD}\big((\succeq_i)_{i\in N},\sigma\big)$ is Pareto optimal. For the convenience of our readers, we include a proof in Appendix~\ref{app_SD}. 
 
 \smallskip
The following definition extends the concept of Pareto optimality to lotteries over outcomes.
\begin{definition}A lottery $p\in \Delta(X)$ is \df{ex-post Pareto optimal}, or \df{ex-post efficient}, if it is supported on Pareto optimal outcomes; i.e., $p_x>0$ implies that $x$ is Pareto optimal.
\end{definition}

If the ordering of agents $\sigma$ in the definition of a serial dictatorship is taken uniformly at random as well as the outcome $x\in \mathrm{SD}\big((\succeq_i)_{i\in N},\sigma\big)$, we obtain the celebrated Random Serial Dictatorship mechanism (RSD). Formally, let $S_N$ be the set of permutations of $N$. RSD outputs the following lottery over outcomes:
$$\mathrm{RSD}\big((\succeq_i)_{i\in N}\big)=\frac{1}{n!}\sum_{\sigma\in S_n}\mathrm{Uniform}\big(\mathrm{SD}((\succeq_i)_{i\in N},\sigma)\big),$$
where $\mathrm{Uniform}(Y)$ for $Y\subseteq X$ denotes the uniform lottery over $Y$. 
By definition, $\mathrm{RSD}\big((\succeq_i)_{i\in N}\big)$ is ex-post efficient.

\medskip
Now we turn to an alternative, stronger, notion of efficiency for lotteries. Agent $i$'s preference $\succeq_i$  over $X$ induces an incomplete preference over lotteries $\Delta(X)$, called  (first-order) stochastic dominance: $p$ stochastically dominates (SD) $q$ for preference $\succeq_i$ if $p$ puts at least as much weight as $q$ on all upper-contour sets of the preference relation $\succeq_i$. In an abuse of notation, we use $\succeq_i$ to denote the SD dominance relation over lotteries, so that for $p, q\in \Delta(X)$, 
$$p\succeq_i q\quad \Leftrightarrow  \quad \sum_{y\succ_i x} p_y\geq \sum_{y\succ_i x} q_y\quad \forall x\in X.$$ If at least one of these inequalities is strict, then we write $p\succ_i q$.
\begin{definition}
A lottery $p\in \Delta(X)$ is \df{ex-ante efficient} (or stochastic-dominance efficient), if there is no lottery $q\in \Delta(X)$ such that $q\succeq_i p$ for all agents $i\in N$ and $q\succ_j p$ for some $j$.
\end{definition}

Ex-ante efficiency implies its ex-post cousin. Indeed, if $p$ is not ex-post efficient, we can construct a stochastically dominating lottery by moving weight from dominated to dominating outcomes.

\section{Geometry and combinatorics of ex-ante efficiency}\label{sec_welfare}

A natural way to obtain an ex-ante efficient lottery is to represent each agents' ordinal preference by a utility function, and then choose a lottery that maximizes the sum of expected utilities. {We need the converse statement proved by \cite{carroll2010efficiency}: all ex-ante efficient outcomes can be found in this way. 
}

Recall that a function $u_i:\ X\to \R$ is a utility function representing agent $i$'s preference $\succeq_i$ if $u_i(x)\geq u_i(y)$ if and only if $x\succeq_i y$. Equivalently $u_i(x)>u_i(y)$ if $i$ strictly prefers $x$ to $y$ and $u_i(x)=u_i(y)$ if $i$ is indifferent. A collection of $(u_i)_{i\in N}$ is a utilitarian representation of a profile of preferences $(\succeq_i)_{i\in N}$ if $u_i$ represents $\succeq_i$ for each agent $i\in N$.

\begin{theorem}[\cite{carroll2010efficiency}]\label{lm_welfarist}
A lottery $p$ is ex-ante efficient if and only if there are utilitarian representations $(u_i)_{i\in N}$ of agents' preferences such that $p$ is supported on outcomes $x$ with maximal welfare $\sum_{i\in N} u_i(x)$. Formally,
$$p\left(\arg\max_{x\in X} \sum_{i} u_i(x)\right) =1.$$
\end{theorem}
{For reader's convenience, a simple proof  of Theorem~\ref{lm_welfarist} based on a straightforward application of linear programming duality is included in Appendix~\ref{app_welfare}.} 
Theorem~\ref{lm_welfarist} implies several insights regarding the structure of ex-ante efficient allocations {that can also be found in \cite{aziz2015universal}.} 
Denote the support of a lottery $p$ by $\supp [p]=\{x\in X\, :\, p_x>0\}$.

\begin{corollary}[\cite{aziz2015universal}]\label{cor_ex_ante_support}
For any ex-ante efficient lottery $p$, any lottery $q$ such that $\supp [q]\subseteq \supp [p]$ is also ex-ante efficient.
\end{corollary}
Indeed, the utilitarian representation $(u_i)_{i\in N}$ from Theorem~\ref{lm_welfarist} constructed for $p$ can also be used for $q$.
Corollary~\ref{cor_ex_ante_support} implies that, whether a lottery is ex-ante efficient, is determined by its support.  \cite{abdulkadirouglu2003ordinal} prove a similar result for the housing market. 
The next corollary provides a geometric description of ex-ante efficient lotteries. 
Consider a collection of subsets $\mathcal{X}\big((\succ_i)_{i\in N}\big)\subset 2^X$ such that $S \in \mathcal{X}$ if and only if there is an ex-ante efficient lottery $p$ with $\supp[p]=S$.  
\begin{corollary}[\cite{aziz2015universal}]\label{cor_faces} The set of all ex-ante efficient lotteries coincides with the union of simplices $$\bigcup_{S\in \mathcal{X}\big((\succ_i)_{i\in N}\big)}\Delta(S).$$
\end{corollary}
Corollary~\ref{cor_faces}  admits a geometric interpretation. By Theorem~\ref{lm_welfarist}, ex-ante efficient lotteries are those elements of the simplex $\Delta(X)$  of all lotteries, a polytope, that can be obtained via maximization of welfare, which is a linear functional. Maximizing a linear functional over a polytope, we always obtain a face of the polytope. The sets $\Delta(S)$ from Corollary~\ref{cor_faces} are exactly those faces that can be obtained in this way, as the linear functional varies over the cone of welfare functionals corresponding to utilitarian representations of a given preference profile.

Ex-ante efficiency implies ex-post, so any $S\in \mathcal{X}$ must consist of Pareto optimal outcomes, i.e., $\mathcal{X}\subset 2^{X^*}$. On the other hand, for any Pareto-optimal outcome $x\in X^*$, the degenerate lottery  $\delta_{x}$, placing all weight on $x$, is ex-ante efficient. Thus  $\{x\}\in \mathcal{X}$ for any Pareto optimal $x$. 
In general, whether a set $S$ belongs to $\mathcal{X}$ can be determined via a simple linear system. 
\begin{lemma}\label{lm_exante_support_characterization}
A set $S\subset X^*$ supports an ex-ante efficient lottery if and only if
there is no $\alpha\in \Re^{S}\times \Re_+^{X^*\setminus S}$ such that $\sum_{x\in X^*} \alpha_x=0$, 
\begin{equation}\label{eq_alpha_general}
\sum_{y\succ_i x} \alpha_y\geq 0 \qquad \mbox{for all $i\in N$ and $x\in X^*$}
\end{equation}
and at least one of these inequalities is strict. 
\end{lemma}
{A similar system for checking ex-ante efficiency of a given lottery is contained in \citep{aziz2015universal}. The problem they study is equivalent because, by Corollary~\ref{cor_ex_ante_support}, checking that $S$ supports an ex-ante efficient lottery is equivalent to checking ex-ante efficiency of the uniform lottery over $S$. }
\begin{proof}[Proof of Lemma~\ref{lm_exante_support_characterization}]
By Corollary~\ref{cor_ex_ante_support}, a set $S$ supports an ex-ante efficient lottery if and only if the uniform lottery $p$ on $S$ is ex-ante efficient. If $p$ is dominated by some $q$, we obtain $\alpha$ solving the linear system by setting $\alpha=q-p$. Conversely, if $\alpha$ is a solution, $q=p+\varepsilon\alpha$ is a feasible lottery, for small enough positive $\varepsilon$, and $q\succeq_i p$ for all agents $i$ while $q\succ_i p$ for some $i$. 
\end{proof}
Lemma~\ref{lm_exante_support_characterization} admits a combinatorial reformulation.
\begin{lemma}\label{lm_ala_Fishburn}
A set $S\subset X^*$ supports an ex-ante efficient lottery if and only if for any $K\geq 1$ we cannot find a pair of sequences $(a_k)_{k=1,\ldots, K}$ consisting 
of outcomes from $X^*$ and $(b_k)_{k=1,\ldots, K}$ consisting 
of outcomes from $S$ and  such that for any outcome $x\in X^*$  and all agents $i$
\begin{equation}\label{eq_Fishburn}
\left|\{k\,:\, a_k\succ_i x \}  \right|\geq \left|\{k\,:\, b_k\succ_i x \}  \right|
\end{equation}
and for at least one $i$ and $x$, the inequality is strict.
Repetitions are allowed, but each outcome enters at most one of the sequences and is repeated there  at most 
$\exp\left(\frac{1}{2}|X^*|\cdot\ln |X^*|\right)$ times.
\end{lemma}
The lemma is proved in Section~\ref{sec_Fishburn_proof} via integer linear programming tools from Appendix~\ref{app_integral}. A related condition was found by~\cite{fishburn1969preferences}. Pick a utilitarian representation of a given preference profile and order the outcomes by their welfare. \cite{fishburn1969preferences} calls partial orders that can be obtained if we are free to choose the utilitarian representation ``summable'' and characterizes all summable orders. By Lemma~\ref{lm_welfarist}, a set $S$ supports an ex-ante efficient lottery if and only if $S$ is the top indifference class for some summable partial order. This connection explains the similarity between the condition from Lemma~\ref{lm_ala_Fishburn} and the condition (4b) of \cite{fishburn1969preferences}. {A characterization of ex-ante efficiency via multi-sets of outcomes similar to Lermma~\ref{lm_ala_Fishburn} was obtained by \cite{aziz2014characterization}.
In contrast to Lemma~\ref{lm_ala_Fishburn}, \cite{fishburn1969preferences} and \cite{aziz2014characterization} did not bound the number of repetitions.}

\section{General results on the equivalence of ex-ante and ex-post efficiency}\label{sec_criteria_of_equivalence}
{Insights about the structure of ex-ante efficient allocations from Section~\ref{sec_welfare} give general criteria for ex-ante and ex-post efficiency to coincide at a given preference profile. In subsequent sections these criteria will be applied and refined for particular  domains of preferences.}
\begin{criterion}\label{crit_main}
For a profile of preferences $(\succeq_i)_{i\in N}$ the following assertions are equivalent:
\begin{itemize}
    \item[(a)] The set of ex-ante efficient lotteries coincides with the set of ex-post efficient lotteries.
     \item[(b)] The set of ex-ante efficient lotteries is convex.
    
      \item[(c)] The uniform lottery over the set $X^*$ of Pareto-optimal outcomes is ex-ante efficient.
   
    \item[(d)] There is a utilitarian representation  $(u_i)_{i\in N}$ such that all Pareto-optimal utilities $x$ have the same welfare $\sum_{i} u_i(x)$. 
\item[(e)] There is no $\alpha\in \Re^{X^*}$ such that $\sum_{x\in X^*} \alpha_x=0$, 
\begin{equation}\label{eq_criterion_alpha}
\sum_{y\succ_i x} \alpha_y\geq 0 \qquad \mbox{for all $i\in N$ and $x\in X^*$}
\end{equation}
and at least one of these inequalities is strict.
\item[(f)] We cannot find a pair of sequences $(a_k)_{k=1,\ldots, K}$  and $(b_k)_{k=1,\ldots, K}$ consisting 
of outcomes from $X^*$ such that each outcome is repeated at most
$\exp\left(\frac{1}{2}|X^*|\cdot\ln |X^*|\right)$ times and 
\begin{equation}\label{eq_Fishburn_criterion}
\left|\{k\,:\, a_k\succ_i x \}  \right|\geq \left|\{k\,:\, b_k\succ_i x \}  \right|
\end{equation}
for any $x\in X^*$  and $i$ and for at least one combination, the inequality is strict.
\end{itemize}
\end{criterion}
Criterion $(b)$ was alluded to in the introduction. If $p$ and $p'$ are ex-ante efficient, then agents may disagree on how the two lotteries should be ranked. Fairness would then suggest the lottery $\frac{1}{2}p+\frac{1}{2}p'$ as a compromise. Criterion $(b)$ means that the ex-ante efficiency of such compromise lotteries characterize the environments in which ex-ante and ex-post efficiency coincide. 
\begin{proof}
Let us demonstrate that $(a)\Leftrightarrow (b)$. If set of ex-post efficient lotteries $\Delta(X^*)$ coincides with the set of ex-ante efficient lotteries, then the latter is convex since $\Delta(X^*)$ is convex, i.e., $(a)\Rightarrow (b)$.
The set of ex-ante efficient lotteries contains a point mass at $x$ for each Pareto optimal $x$. Hence, if the  set of ex-ante efficient lotteries is convex, it contains any convex combination of such distributions, i.e., it contains $\Delta(X^*)$ and thus equals $\Delta(X^*)$. We deduce that  $(b)\Rightarrow (a)$. The equivalence $(a)\Leftrightarrow (b)$ is established.

The equivalence $(a)\Leftrightarrow (c)$ follows from Corollary~\ref{cor_ex_ante_support}. By this corollary, the set of ex-post efficient lotteries $\Delta(X^*)$ coincides with the set of ex-ante efficient lotteries if and only if $X^*$ enters  the decomposition from Corollary~\ref{cor_ex_ante_support}, i.e., $X_k^*=X^*$ for some $k$. The set $X^*$ enters the decomposition if and only if there is an ex-ante efficient lottery $p$ with $\supp p=X^*$. Choosing $p$ to be the uniform lottery over $X^*$, we conclude that $(a)\Leftrightarrow(c)$. 

Now we check that $(c)\Leftrightarrow (d)$. 
First, let us show that $(c)\Rightarrow (c)$. If the uniform lottery $p$ over $X^*$ is ex-ante efficient, then, by Lemma~\ref{lm_welfarist}, there is a utilitarian representation $(u_i)_{i\in N}$ of preferences such that $p$ is supported on outcomes $x$ maximizing the welfare $\sum_i u_i(x)$. In particular, for all the outcomes from the support, the welfare is the same.
Hence, $(c)$ implies $(d)$. To show that $(d)\Rightarrow(c)$, 
we start with a utilitarian representation $(u_i)_{i\in N}$, where all the Pareto optimal outcomes $x\in X^*$ have the same welfare of $w$ and show that the uniform lottery $p$ over $X^*$ is welfare-maximizing, i.e., the condition of Lemma~\ref{lm_welfarist} holds.  For any outcome $y\in X$, its welfare cannot be higher than $w$ as this would contradict Pareto optimality of $X^*$. Hence, any lottery with welfare of $w$ is a welfare maximizer. Since the uniform lottery has welfare $w$, we conclude that it maximizes welfare and so is ex-ante efficient.

The equivalence $(c)\Leftrightarrow (e)$ is an immediate corollary of Lemma~\ref{lm_exante_support_characterization} with $S=X^*$.

Similarly, the equivalence $(c)\Leftrightarrow (f)$ follows from Lemma~\ref{lm_ala_Fishburn}.
\end{proof}
\subsubsection{Applications} {We illustrate a use of criterion $(e)$ and recover the result by \cite{aziz2018tradeoff} that ex-ante and ex-post efficiency coincide in any preference domain provided that the number of agents and/or Pareto optimal outcomes are small enough.}
\begin{proposition}[\cite{aziz2018tradeoff}]\label{prop_small_profiles_general}
The sets of ex-ante and ex-post efficient lotteries coincides for any preference profile $(\succeq_i)_{i\in N}$
\begin{enumerate}
    \item with $|X^*|\leq 3$ Pareto optimal outcomes.
    \item with $|N|\leq 2$  agents.
    \item with $|N|=3$ agents and $|X^*|\leq 5$ Pareto optimal outcomes.
\end{enumerate}
\end{proposition}
\cite{aziz2018tradeoff} also demonstrated that the result is tight, namely, for any pair of numbers $|N|$ and $|X^*|$ that do not match any of the above three cases, there is a profile of preferences where ex-ante and ex-post efficiency do not coincide. Such profiles can also be found in Corollary~\ref{cor_tightness} of Section~\ref{sec:socialchoice}.

The proof of Proposition~\ref{prop_small_profiles_general} uses the following invariance of the set of preference profiles, where ex-ante efficiency is not equivalent to ex-post. Call a profile $(\succeq_i')_{i\in N}$ a reversal of $(\succeq_i)_{i\in N}$ if each agent preferences are reversed, i.e., $x\succ_i' y$ if and only if $x\prec_i y$.
\begin{lemma}\label{lm_dual}
Consider a profile of preferences $(\succeq_i)_{i\in N}$ such that all outcomes are Pareto optimal. Then the sets of ex-ante and ex-post efficient lotteries are distinct for $(\succeq_i)_{i\in N}$ if and only if they are distinct for its reversal $(\succeq_i')_{i\in N}$.
\end{lemma}
\begin{proof}[Proof of Lemma~\ref{lm_dual}:]
Since all outcomes are Pareto optimal for $(\succeq_i)_{i\in N}$, whenever there is an agent strictly preferring $x$ to $y$, there is also one strictly preferring $y$ to $x$. Hence, all outcomes are Pareto optimal for the reversal $(\succeq_i')_{i\in N}$ as well. By criterion $(e)$, ex-ante and ex-post efficiency are distinct for $(\succeq_i)_{i\in N}$ if and only if we can find $\alpha\in\R^{X^*}$ solving the system~\eqref{eq_criterion_alpha}. Whenever such $\alpha$ exists, $\alpha'=-\alpha$ solves the corresponding system for the reversal $(\succeq_i')_{i\in N}$. The two systems admit a solution simultaneously, and  thus ex-ante and ex-post efficiency are distinct for the two profiles simultaneously as well. 
\end{proof}

\begin{proof}[Proof of Proposition~\ref{prop_small_profiles_general}]
We prove each of the claims separately:\\
\smallskip
\noindent (1) Towards a contradiction, assume that there is a profile with at most $3$ Pareto optimal outcomes, where ex-ante and ex-post efficiency are distinct. Hence, we can find  $\alpha\in \R^{X^*}$ solving the system from criterion~$(e)$. Possibly eliminating some outcomes, we can assume that, without loss of generality, $\alpha_x\ne 0$ for any outcome $x\in X$ and $X=X^*$. As $|X^*|\leq 3 $, the vector $\alpha$ has either only one negative component or only one positive (or both).
 
First, consider the case of only one $b\in X^*$ with $\alpha_b<0$. 
By~\eqref{eq_criterion_alpha}, the outcome $b$ must belong to the lowest indifference class of each agent, i.e., any outcome is weakly preferred to $b$ by all the agents. The agent $i$  with strict  inequality in~\eqref{eq_criterion_alpha} has at least two indifference classes and so there is an outcome strictly preferred by $i$ to~$b$. We obtain a contradiction with Pareto optimality of $b$. Now, consider the case  of only one $a\in X^*$ with $\alpha_a>0$. In the dual profile  from Lemma~\ref{lm_dual}, $a$ will be the only outcome with the negative weight and we again get a contradiction

\smallskip    
\noindent (2)  Towards a contradiction, assume that there exists a profile of preferences $(\succeq_i)_{i\in N}$ with  $N=\{1,2\}$ such that ex-ante and ex-post efficiency are distinct and let $\alpha$ be the corresponding solution to the system~\eqref{eq_criterion_alpha} from criterion~$(e)$.
Without loss of generality, all outcomes are Pareto optimal and, hence, the preferences of agents $1$ and $2$ must be opposite, i.e., $x\succ_1 y$ whenever $x\prec_2 y$. By the criterion, we know that there is an agent $i$, say $i=1$, and an outcome $x$ such that 
$\sum_{y\succ_1 x } \alpha_y>0$. Since an upper contour of agent $1$'s preference is a lower contour of agent $2$'s preference and $\sum_{y\in X} \alpha_y=0$, we conclude that there is $x'$ such that $\sum_{y\succ_2 x' } \alpha_y<0$, which contradicts the assumption that $\alpha$ solves~\eqref{eq_criterion_alpha}.  

\smallskip
\noindent (3) Towards a contradiction, assume that there is a profile $(\succeq_i)_{i\in N}$ with $N=\{1,2,3\}$, $|X^*|\leq 5$, and such that ex-ante and ex-post efficiency are distinct. Consider  a solution $\alpha$ to the system~\eqref{eq_criterion_alpha} from criterion~$(e)$. Without loss of generality, all outcomes are Pareto optimal, $\alpha_x\ne 0$ for all  $x$, and there are no distinct outcomes  such that all agents are indifferent between them. 
Since there are at most $5$ outcomes, there are either at most two outcomes $x$ with $\alpha_x<0$ or at most two with $\alpha_x>0$ (or both). 

Consider the first scenario of at most two  outcomes $x$ with $\alpha_x<0$. In the proof of item~(1), we saw that there must be at least two such outcomes. Thus there are exactly two of them  which we denote by $b_1$ and $b_2$. By~\eqref{eq_criterion_alpha}, each agent's least preferred indifference class must contain at least one of $b_1$ or $b_2$. Let  $b\in\{b_1,b_2\}$ be the outcome appearing in the least-preferred class of two agents, say, agents $i\in \{1,2\}$. To be Pareto optimal, $b$ must belong to the top indifference class of the remaining agent~$3$. Let us demonstrate that the latter indifference class must be a singleton. Indeed, if it contains some other outcome $a$, then, since $b$ cannot be  dominated by $a$, the outcome $a$ must belong to the  least-preferred indifference classes of agents~$1$ and~$2$, which is ruled out by the assumption that all agents cannot be indifferent between a pair of outcomes. We conclude that the outcome $b$ with $\alpha_b<0$ is the only top-ranked outcome for agent~$3$ which contradicts the inequality~\eqref{eq_criterion_alpha} requiring that all upper contours have non-negative sum of $\alpha$-weights.

In the second scenario, there are at most two  outcomes $x$ with $\alpha_x>0$. In the dual profile from Lemma~\ref{lm_dual}, the same two outcomes get strictly negative weight and the above argument leads to contradiction.
\end{proof}

Criteria $(d)$ and $(e)$ have the form of a linear system, and so there is an efficient algorithm for checking whether ex-ante and ex-post efficiency coincide at a given preference profile. We spell out the linear system for~$(e)$: 
\begin{enumerate}
    \item First, one precomputes the set of Pareto optimal outcomes $X^*$: an outcome $x$ belongs to $X^*$ if for any other $y\in X$, either all agents are indifferent between $x$ and $y$ or there is an agent $i$ who prefers $x$ to $y$. Enumerating all the combinations $x$, $y$, and $i$ requires a polynomial number $O(|X|^2\cdot |N|)$ of operations.
    \item The linear system $(e)$ is homogeneous. Hence, the existence of a solution $\alpha$ such that $\sum_{y\succ_i x} \alpha_y>0$ for some $i=i^*$ and $x=x^*$ is equivalent to the existence of a solution such that $\sum_{y\succ_{i^*} x^*} \alpha_y\geq 1$.
    Thus the condition $(e)$ is equivalent to infeasibility of a family of linear systems indexed by $i^*\in N$ and $x^*\in X^*$:
    $$\left\{\begin{array}{ccc}\sum_{x\in X^*} \alpha_x&=0 & \\ 
\sum_{y\succ_i x} \alpha_y,&\geq 0 & i\in N,\, x\in X^*\\
\sum_{y\succ_{i} x} \alpha_y,&\geq 1 &  i=i^*,\, x=x^*
\end{array}\right.$$
\end{enumerate}
We conclude that checking that ex-ante efficiency coincides with ex-post efficiency boils down to checking the feasibility of $O(|N|\cdot|X^*|)$ linear systems of dimension $O(|X^*|)$, with $O(|N|\cdot|X^*|)$ constraints.
\begin{corollary}\label{cor_algo}
Given a profile of preferences $(\succeq_i)_{i\in N}$, whether or not the set of ex-ante efficient lotteries coincides with the set of ex-post efficient lotteries can be decided in polynomial time.
\end{corollary}

\section{Dichotomous preferences. }\label{sec:dichotomous}

We consider preference profiles $(\succeq_i)_{i\in N}$ of dichotomous preferences: each agent $i\in N$ regards a subset of the outcomes $X$ as acceptable, but makes no distinction among acceptable outcomes (each $\succeq_i$ features only two indifference classes). 
{The following result by \cite{duddy2015fair} refines Proposition~\ref{prop_small_profiles_general} in the dichotomous domain.}
\begin{proposition}[\cite{duddy2015fair}]\label{prop:basicdichotomous}
If $|N|\leq 4$, a lottery is ex-ante efficient if and only if it is ex-post efficient. This is not true when $|N|\geq 5$.
\end{proposition}
{A minimal example with $5$ agents and $4$ outcomes were not all ex-post efficient lotteries are ex-ante efficient can be found in the introduction to our paper or in \citep{aziz2019fair}.}
{In Section~\ref{sec_proof_minimal_dichotomous}, we show how to prove Proposition~\ref{prop:basicdichotomous} using the technique from Section~\ref{sec_criteria_of_equivalence}.} Note that Proposition~\ref{prop:basicdichotomous}  classifies environments purely in terms of the size of $N$. The reason is that in the dichotomous domain, the number of non-equivalent outcomes can be bounded in terms of $|N|$.

Next, we ask for which profiles (with $|N|\geq 5$ agents) is ex-ante efficiency  equivalent to ex-post. In addition to the criteria from Section~\ref{sec_criteria_of_equivalence}, which can be refined and specialized to the dichotomous domain, we obtain a new domain-specific criterion featuring the Random Serial Dictatorship.

It will be convenient to identify each outcome $x$ with a vector $x\in \{0,1\}^N$ such that $x_i=1$ if $x$ is acceptable for $i$ and $x_i=0$, otherwise. A canonical utilitarian representation of $(\succeq_i)_{i\in N}$ can be obtained by putting $u_i(x)=x_i$. 

\begin{proposition}\label{prop:characterizdichotomous} Let $(\succeq_i)_{i\in N}$ be a profile of dichotomous preferences. The following are equivalent:
\begin{enumerate}
    \item All ex-post efficient lotteries are ex-ante efficient.
    \item $\mathrm{RSD}\big((\succeq_i)_{i\in N}\big)$ is ex-ante efficient.
    \item  \label{item_lamda_criterion}  There is a vector $\lambda\in \Re^N$ with $\lambda_i>0$ for all $i$ such that the dot product $x\cdot \lambda$ is the same for all Pareto optimal outcomes $x$.
    \item\label{it:dichseq} There are no two sequences of Pareto optimal outcomes 
    $(a_k)_{k=1,\ldots, K}$ and $(b_k)_{k=1,\ldots, K}$ with no terms in common such that
    \begin{equation}\label{eq_Fishburn_dichotomous}
\sum_{k=1}^K a_k \geq \sum_{k=1}^K b_k
\end{equation}
 and such that the vector at the left-hand side is not equal to the vector at the right-hand side. Moreover, there is a bound $B_{n}$ that only depends on $n=|N|$ (and not on $X$ or agents' preferences), so that it is enough to check sequences where each outcome occurs at most $B_n$ times.
\end{enumerate}
\end{proposition}

The proposition provides some simple examples of dichotomous domains where ex-ante and ex-post efficiency coincide. 
\begin{example}\label{ex_dichotomous_equisupport}
If all outcomes are viewed as acceptable by the same number of agents, then by taking $\lambda=(1,1,\ldots, 1)$ in item~\ref{item_lamda_criterion} we conclude that ex-ante and ex-post efficiency are equivalent for such a profile. Similarly, if outcomes and agents can be partitioned as $X=\cup_{s=1}^m X_s$ and $N=\cup_{s=1}^m N_s$ so that $N_s$ are those agents (and only those) who regard $X_s$ as acceptable, then we can take $\lambda$ such that $\lambda_i=\frac{1}{|N_s|}$ for $i\in N_s$.
\end{example}

\section{Single-peaked preferences}\label{sec_singlepeak}

For the dichotomous and house-allocation domains, there are preference profiles where the requirement of ex-ante efficiency is strictly stronger than that of ex-post efficiency. By contrast, single-peaked preferences considered in this section provide an example of a domain where the two efficiency notions are equivalent for any profile.

The single-peaked domain consists of the situations where the set of outcomes $X$ is one-dimensional, each agent has an ideal point, and agents rank an outcome over another when it is closer to their ideal point. The one-dimensional structure is captured by a strict ordering $\prec$ over outcomes. A preference $\succeq_i$ of an agent $i$ is said to be single-peaked with respect to $\prec$ if there is an outcome $x_i^*\in X$ such that $x_i^*\preceq x\prec y$ or $y\prec x\preceq x_i^*$ implies  $x\succ_i y$. A profile $(\succeq_i)_{i\in N}$ is a profile of single-peaked preferences if there exists a common ordering $\prec$ of outcomes such that each $\succeq_i$ is single-peaked with respect to~$\succ$.
\begin{proposition}\label{prop_singlepeak}
For any profile of single-peaked preferences, the set of ex-ante efficient lotteries coincides with the set of ex-post efficient ones.
\end{proposition}
The proof of Proposition~\ref{prop_singlepeak} relies on Criterion~\ref{crit_main}~(e), and may be found in Section~\ref{sec_singlepeak_proof}.
{A particular case of Proposition~\ref{prop_singlepeak} for   single-peaked dichotomous preferences is contained in \citep{aziz2019fair}.}

\section{Social choice domain}\label{sec:socialchoice}

In the housing model, agents are indifferent over outcomes that do not differ in their private consumption. For the dichotomous domain, agents only have two indifference classes. In contrast, we now turn to environments where agents' preferences are strict; no indifferences are allowed, but preferences may otherwise be arbitrary. The resulting class of environments are termed the \df{social choice domain}.

Let $N$ be a finite set of agents, $X$ a finite set of outcomes, and  $(\succeq_i)_{i\in N}$  a profile of preferences over $X$ which are now assumed to be strict.

In contrast to the single-peaked domain considered in Section~\ref{sec_singlepeak}, there are profiles where not all ex-post efficient lotteries are ex-ante efficient; see Example~\ref{ex_2Condorcet} in the introduction.
{However, as it was mentioned by~\cite{brandl2016impossibility},}  even for such profiles, the Random Serial Dictatorship (RSD) always outputs an ex-ante efficient lottery.
\begin{proposition}\label{prop_RSD_social_choice}
For any  profile of strict preferences $(\succeq_i)_{i\in N}$ the $\mathrm{RSD}\big((\succeq_i)_{i\in N}\big)$
is ex-ante efficient.
\end{proposition}
This property of the social choice domain is to be contrasted with the phenomenon that we observed in the dichotomous domain, where ex-ante efficiency of RSD determines whether or not the sets of ex-ante and ex-post efficient lotteries coincide.
\begin{proof}
Let $p=\mathrm{RSD}\big((\succeq_i)_{i\in N}\big)$. Then $p$ is supported on outcomes $x$ such that there is an agent $i$ for whom $x$ is the most preferred outcome. Consider a lottery $p'\ne p$. Let $y$ be an outcome such that $p_y> p_y'$ and $i$ be an agent that top-ranks $y$. Since $i$ has only one most preferred outcome, the chance that $i$ receives her most-preferred outcome is strictly lower under $p'$ than under $p$. Thus no lottery $p'$ can dominate the outcome of RSD, i.e.,  this outcome is ex-ante efficient.
\end{proof}

\begin{remark}
Note that the proof of Proposition~\ref{prop_RSD_social_choice} only relies on the fact that each agent's top choice is unique. In particular, RSD is ex-ante efficient in any domain where the top indifference class of each agent is a singleton.
\end{remark}

To construct preference profiles where ex-ante and ex-post efficiency are not equivalent, one can use a connection between Criterion~\ref{crit_main}~(f) and the Bertrand ballot-counting problem from probability theory.\footnote{We are grateful to Ron Holzman who suggested this connection.} In the classic version of the ballot-counting problem, there are two candidates, $A$ and $B$, and we are given a ballot box with voting slips in equal numbers for the two candidates. We count the votes in random order, and the basic question is to find the probability that, at every stage of the count, the current count in favor of $A$ is at least the current count in favor of~$B$.

Here we consider a weighted version of the ballot problem. The ballot box is stuffed with envelopes, each containing multiple slips in favor of one candidate. We identify the set of outcomes $X$ with the set of envelopes in the box, and the set of agents $N$ with a set of permutations of $X$. So let  $X$ be the collection of envelopes, each of which contains some positive integer number of slips for one of the two candidates. Suppose that the total number of slips for $A$ equals that for $B$. Let $N$ be a family of permutations of $X$ having two properties: 
\begin{itemize}
    \item[i)] For every permutation in $N$, at every stage, the current count for $A$ (the total number of slips in the envelopes counted so far) is at least that for $B$.
    \item[ii)] For any two envelopes $x$ and $x'$, there is a permutation in $N$ such that $x$ appears before $x'$, and another one in which $x'$ appears before $x$.
\end{itemize}
We interpret envelopes $X$ as outcomes, and permutations $N$ as agents. The (strict) preference of an agent are such that her ordering of outcomes, from best to worst, is given by the corresponding permutation.  We call the resulting profile of preferences a \emph{ballot-counting profile}.
 
For ballot-counting profiles, ex-ante  and ex-post efficiency turn out to be distinct. Moreover, such profiles play the role of \emph{forbidden structures}: meaning that, for any profile of strict preferences, ex-ante and ex-post efficiency coincide whenever there is no ballot-counting sub-profile.
 
Two preference profiles are said to be equivalent if there is a way to rename the agents and the outcomes in one of them to obtain the other. The following proposition is proved in Section~\ref{sec_proof_ballot}.
 \begin{proposition}\label{prop_ballot_and_efficiency}
For a profile of strict preferences $(\succ_i)_{i\in N}$ over $X$, the set of ex-post efficient lotteries coincides with the set of ex-ante efficient lotteries, if and only if there is no subset $Y$ of Pareto optimal outcomes such that $(\succ_i)_{i\in N}$ restricted to $Y$ is equivalent to a ballot-counting profile.
 \end{proposition}
 \begin{example}
 The profile of preferences from Example~\ref{ex_2Condorcet} is equivalent to a ballot-counting profile. We need $3$ envelopes $A_1$, $A_2$, and $A_3$ each with one slip for candidate $A$ and $3$ envelopes $B_1$, $B_2$, and $B_3$ with one slip for candidate $B$. The permutations $N$ are as follows:
 \begin{align*}
     A_1 B_1 A_2 B_2 A_3 B_3\\
A_2 B_3 A_3 B_1 A_1 B_2\\
A_3 B_2 A_1 B_3 A_2 B_1
 \end{align*}
i.e., $A$ and $B$ alternate and we cyclically permute $A_1, A_2, A_3$ and $B_1, B_2, B_3$  in the opposite directions.
\end{example}
One can construct profiles with less than six outcomes where not all ex-post efficient lotteries are ex-ante efficient at the cost of increasing the number of agents. The minimal number of outcomes that we need is $4$ and it requires $4$ agents.
\begin{example}
Consider $2$ envelopes $A_1$ and $A_2$ with one slip for candidate $A$ and $2$ envelopes $B_1$ and $B_2$ with one slip for candidate $B$. The preferences are defined by the four permutations alternating between $A$ and $B$:
 \begin{align*}
     A_1 B_1 A_2 B_2\\
     A_2 B_1 A_1 B_2\\
     A_1 B_2 A_2 B_1\\
     A_2 B_2 A_1 B_1
 \end{align*}

\end{example}
These examples demonstrate tightness of Proposition~\ref{prop_small_profiles_general}.
\begin{corollary}[\cite{aziz2018tradeoff}]\label{cor_tightness}
Proposition~\ref{prop_small_profiles_general} is tight. Namely, for any $n$ and $m$ such that either 
\begin{itemize}
    \item $n\geq 3$ and $m\geq 6$ or
    \item $n\geq 4$ and $m\geq 4$
\end{itemize}
there exists a profile of preferences with $|N|=n$ agents and $|X^*|=m$ Pareto optimal outcomes such that the sets of ex-ante and ex-post efficient lotteries are distinct.
\end{corollary}

Let us call a ballot-counting profile simple if each envelope contains exactly one slip for $A$ or $B$. In all the examples we saw so far, the constructed profiles were simple. Do non-simple profiles exist? Simple profiles correspond to $\alpha_x\in \{-1,0,1\}$ in Criterion~\ref{crit_main}~(d). By Lemma~\ref{lm_integrality} from Appendix~\ref{app_integral},
without loss of generality, $\alpha$ has integer coordinates bounded by $\exp\left(\frac{1}{2}|X^*|\cdot \ln |X^*|\right)$ in absolute value. Can this bound be replaced by $1$?
The answer is negative as the following example demonstrates the existence of non-simple profiles.
\begin{example}
Envelopes $A_1$, $A_2$, $A_3$, and  $A_4$ contain one slip and $B_1$ and $B_2$, two slips. The set $N$ contains $12$ permutations all of which have the form $AABAAB$.  Each unordered pair $A_i$, $A_j$ appears twice as the first pair of $A$'s, in one of them $B_1$ precedes $B_2$ and in the other $B_2$ precedes $B_1$:
 \begin{align*}
     A_1 A_2 B_1 A_3 A_4 B_2\\
     A_1 A_3 B_1 A_2 A_4 B_2\\
     A_1 A_4 B_1 A_3 A_2 B_2\\
     A_2 A_3 B_1 A_1 A_4 B_2\\
     A_2 A_4 B_1 A_1 A_3 B_2\\
     A_3 A_4 B_1 A_1 A_2 B_2\\
      A_1 A_2 B_2 A_3 A_4 B_1\\
     A_1 A_3 B_2 A_2 A_4 B_1\\
     A_1 A_4 B_2 A_3 A_2 B_1\\
     A_2 A_3 B_2 A_1 A_4 B_1\\
     A_2 A_4 B_2 A_1 A_3 B_1\\
     A_3 A_4 B_2 A_1 A_2 B_1\\
 \end{align*}

\end{example}

Despite the last example showing that not every ballot-counting profile is simple,  it is enough to consider only simple profiles  if we allow for retractions.

Consider a preference profile $(\succ_i)_{i\in N}$ with the set of outcomes $X$ and assume that a subset of outcomes $A\subset X$ are adjacent in the preferences  of all the agents, i.e., $a\succ_i b\succ_i c$ with $a,c\in A$ implies $b\in A$. Let $(\succ_i')_{i\in N}$ be the preference profile where the set of outcomes $A$ is replaced by one outcome $\{A\}$. We call any profile obtained by a consecutive application of this procedure to different sets of adjacent outcomes a retraction  of the original profile.

Any ballot-counting profile is a retraction of a simple ballot-counting profile. Indeed,
 each envelope with $k>1$ slips can be split into a collection of $k$ envelopes each containing one slip and placed one after another in each of the permutations.

Let us call a retraction trivial if it coincides with the original profile or has just one outcome. A profile is said to be irretractable if it has no non-trivial retractions. A retraction is maximal if it is irretractable. We obtain the following corollary of Proposition~\ref{prop_ballot_and_efficiency}.
\begin{corollary}
For a profile of strict preferences $(\succ_i)_{i\in N}$ over $X$, the set of ex-post efficient lotteries coincides with the set of ex-ante efficient lotteries, if and only if there is no subset $Y$ of Pareto optimal outcomes such that $(\succ_i)_{i\in N}$ restricted to $Y$ is equivalent a maximal retraction of a simple ballot-counting profile.
\end{corollary}

\section{Proofs}\label{sec:proofs}

\subsection{Proof of Lemma~\ref{lm_ala_Fishburn}}\label{sec_Fishburn_proof}
First, we  assume that $S$ does not support an ex-ante efficient lottery and show how to  construct the sequences $(a_k)_{k=1,\ldots, K}$ and $(b_k)_{k=1,\ldots, K}$. By the assumption on $S$, the system from Lemma~\ref{lm_exante_support_characterization} has a solution $\alpha$. By Lemma~\ref{lm_integrality} and the Hadamard inequality~\ref{eq_bound_x_01_matrix} from Appendix~\ref{app_integral}, we can assume that this solution is integral and the absolute value $|\alpha_x|$ is bounded by $\exp\left(\frac{1}{2}|X^*|\cdot\ln |X^*|\right)$ for all $x$. Let $K=\frac{1}{2}\sum_{x \in X^*}|\alpha_x|$ and $(a_k)_{k=1,\ldots, K}$ be the sequence of outcomes $y$ such that $\alpha_y>0$ where each such outcome is repeated $\alpha_y$ times. Similarly, $(b_k)_{k=1,\ldots, K}$ is the sequence of outcomes $y$ with $\alpha_y<0$ and each repeated $|\alpha_y|$ times. Note that each sequence has the same number of outcomes $k$ since $\sum_{x\in X^*}\alpha_x=0$. The inequality~\eqref{eq_alpha_general} becomes equivalent to~\eqref{eq_Fishburn} and, hence, the latter inequality holds and is strict for some $i$ and $x$.

To prove the opposite direction, we assume that the sequences $(a_k)_{k=1,\ldots, K}$ and $(b_k)_{k=1,\ldots, K}$ are given and show that there is no ex-ante efficient lottery supported on $S$. By Lemma~\ref{lm_exante_support_characterization}, it is enough to demonstrate that the system~\eqref{eq_alpha_general} has a solution $\alpha$. We construct it as follows. Put $$\alpha_x=\big|\{k\,\, \colon a_k=x\}\big|-\big|\{k\,\, \colon b_k=x\}\big|.$$
Then inequality~\eqref{eq_Fishburn} is equivalent to~\eqref{eq_alpha_general} and, since the former is strict for some $i$ and $x$, the latter is strict as well.

\subsection{Proof of Proposition~\ref{prop:basicdichotomous}}\label{sec_proof_minimal_dichotomous}
We will need an auxiliary lemma that bounds the number of Pareto optimal outcomes in a profile $(\succeq_i)_{i\in N}$ of dichotomous preferences over $X$. 

Denote by $A_x$ the set of agents that find an outcome $x\in X$ acceptable. We call a pair of distinct outcomes $x$ and $y$ equivalent if $A_x=A_y$. A collection $Y\subseteq X$ consists of non-equivalent outcomes if no pair of outcomes from the collection is equivalent.
\begin{lemma}\label{lm_Sperner}
The number of non-equivalent Pareto optimal outcomes is at most $\binom{|N|}{\left\lfloor |N|/2 \right\rfloor}$. If this upper bound is attained, then either all non-equivalent Pareto optimal outcomes are acceptable for exactly $\left\lfloor |N|/2 \right\rfloor$ agents or all are acceptable for $\left\lceil |N|/2 \right\rceil$ agents.
\end{lemma}
\begin{proof}
Denote by $X_{\not=}^*\subseteq X^*$ a collection of non-equivalent Pareto optimal outcomes. 
For distinct $x,y\in X_{\not=}^*$, the sets $A_x$ and $A_y$ cannot be subsets of each other as this would contradict Pareto optimality. Thus the collection of sets $(A_x)_{x\in X_{\not=}^*}$ forms an  antichain in the power set of $N$. Hence, the size of $X_{\not=}^*$ is bounded by the largest possible cardinality of such an antichain. By the theorem of \cite{sperner1928satz}, the largest cardinality of an antichain of subsets of an $n$-element set is equal to $\binom{n}{\left\lfloor n/2 \right\rfloor}$ and the bound is attained on antichains composed of either all $\left\lfloor n/2 \right\rfloor$-subsets or all $\left\lceil n/2 \right\rceil$-subsets.
\end{proof}

\begin{proof}[Proof of Proposition~\ref{prop:basicdichotomous}]
Consider the following preference profile with $5$ agents $N=\{1,2,3,4,5\}$ and $4$ outcomes $X=\{a,b,c,d\}$ and show that there is an ex-post efficient lottery that is not ex-ante efficient:
\[ 
\begin{array}{c|cccc}
 & a & b & c & d\\
 \hline
1 & 1 & 0 & 1 & 0 \\
2 & 0 & 1 & 0 & 1 \\
3 & 1 & 0 & 0 & 1  \\
4 & 1 & 0 & 0 & 0  \\
5 & 0 & 1 & 1 & 0 \\
\end{array}
\]
Observe that each outcome is Pareto optimal, as the set of agents who deem one acceptable is never a subset of the set that deems another acceptable. Consider two lotteries $p=\left(\frac{1}{2},\frac{1}{2},0,0\right)$ and $q=\left(0,0,\frac{1}{2},\frac{1}{2}\right)$. The lottery $p$ gives each agent an acceptable outcome with probability $\frac{1}{2}$. The lottery $q$ never gives an acceptable outcome to agent $4$, while all other agents get an acceptable outcome with probability $\frac{1}{2}$ and so are indifferent between $q$ and $p$. We conclude that $p\succeq_i q$ for all agents $i$ and $p\succ_4 q$. Thus $q$ is not ex-ante efficient despite that it is ex-post efficient.

If one tries to reduce the number of agents in this construction, then one of the four outcomes becomes Pareto dominated by one of the others. We proceed to prove that, indeed, ex-ante efficiency is the same as ex-post efficiency for $|N|\leq 4$.

Consider a profile of dichotomous preferences $(\succeq_i)_{i\in N}$  with $|N|\leq 4$ agents. 
Towards a contradiction, assume that there  is an ex-post efficient lottery that is not ex-ante efficient. Without loss of generality, all outcomes $X$ are Pareto optimal (i.e., $X=X^*$) and non-equivalent. 

By Proposition~\ref{prop_small_profiles_general}, ex-ante and ex-post efficiency coincide if there are at most three Pareto optimal outcomes. Hence, $|X^*|\geq 4$. By Lemma~\ref{lm_Sperner}, $|X^*|\leq \binom{|N|}{\left\lfloor |N|/2 \right\rfloor}$. This is only possible if $|N|\geq 4$. Thus $|N|=4$ since $|N|$ is assumed to be at most four. It will be convenient to enumerate the agents so that $N=\{1,2,3,4\}$.

For any dichotomous profile where each outcome is acceptable for the same number of agents, ex-ante and ex-post efficiency coincide; see Example~\ref{ex_dichotomous_equisupport}.
Thus outcomes in $X^*$ cannot all have the same number of supporters. No outcome can be supported by all the four agents as it would Pareto dominate all other outcomes; similarly, the case of zero supporters is excluded. We conclude that either there must be an outcome~$d$ acceptable for exactly one agent or for exactly three agents. 

First, consider the case where there is an outcome $d$ acceptable for exactly one agent, say, agent~$4$. Since, $d$ is Pareto optimal and all outcomes are non-equivalent, $d$ must be the only outcome acceptable for agent~$4$. Consider the reduced profile $(\succeq_i)_{i\in N'}$ obtained by eliminating agent~$4$ and her outcome $d$, i.e., $N'=\{1,2,3\}$ and $X'=X\setminus\{d\}$. Hence, $X'$ must consist of Pareto optimal outcomes and contain at least three outcomes as $|X|\geq 4$. By Lemma~\ref{lm_Sperner}, $|X'|\leq 3$ and thus $X'$ contains exactly three outcomes. The same lemma implies that either all these outcomes are acceptable for exactly one agent or all are acceptable for exactly two agents. Therefore, the original profile of preferences must have one of the two possible forms:
\[ 
\begin{array}{c|cccc}
 & a & b & c & d\\
 \hline
1 & 1 & 0 & 0 & 0 \\
2 & 0 & 1 & 0 & 0 \\
3 & 0 & 0 & 1 & 0  \\
4 & 0 & 0 & 0 & 1  
\end{array}
\qquad\mbox{or}\qquad
\begin{array}{c|cccc}
 & a & b & c & d\\
 \hline
1 & 0 & 1 & 1 & 0 \\
2 & 1 & 0 & 1 & 0 \\
3 & 1 & 1 & 0 & 0  \\
4 & 0 & 0 & 0 & 1  
\end{array}.
\]
For both of these profiles ex-ante and ex-post efficiency coincide which follows from item~\ref{item_lamda_criterion} of Proposition~\ref{prop:characterizdichotomous} if we take $\lambda=(1,1,1,1)$ for the first profile and $\lambda=(1,1,1,2)$ for the second.
We conclude that having an outcome $d$ supported by exactly one agent is incompatible with the assumption that ex-ante and ex-post efficiency are distinct.

Now, consider the case where there is an outcome $d$ acceptable for exactly three agents. Recall that in the reversal $(\succeq_i')_{i\in N}$ of $(\succeq_i)_{i\in N}$, acceptable outcomes  are replaced by unacceptable and vice versa for each agent.  By Lemma~\ref{lm_dual}, ex-ante and ex-post efficiency are distinct for the reversal $(\succeq_i')_{i\in N}$. On the other hand, the outcome $d$ is acceptable to exactly one agent in $(\succeq_i')_{i\in N}$ which, as we already showed, is incompatible with the assumption that the sets of ex-ante and ex-post efficient lotteries are not the same. 
This contradiction implies that ex-ante efficiency coincides with ex-post efficiency for any profile of dichotomous preferences with at most four agents.

\end{proof}

\subsection{Proof of Proposition~\ref{prop:characterizdichotomous}}\label{sec:pfchardich}

Let us prove that $(1)\Leftrightarrow(2)$. We know that $\mathrm{RSD}\big((\succeq_i)_{i\in N}\big)$ is ex-post efficient in any domain of preferences (Appendix~\ref{app_SD}). Hence, if the sets of ex-post efficient lotteries coincide with the set of ex-ante efficient lotteries,  $\mathrm{RSD}\big((\succeq_i)_{i\in N}\big)$ is ex-ante efficient, i.e., $(1)\Rightarrow (2)$. To show that $(2)\Rightarrow (1)$, it is enough to check that the lottery $p=\mathrm{RSD}\big((\succeq_i)_{i\in N}\big)$ places a non-zero weight to all the Pareto optimal outcomes $x$ and then apply Corollary~\ref{cor_ex_ante_support}.
Let $A_x$ be the set of agents who deem an outcome $x$ acceptable.
For Pareto optimal outcomes $x$ and $y$, the sets $A_x$ and $A_y$ cannot be proper subsets of each other.  Thus for any Pareto optimal $x$ and any ordering of agents having $A_x$ as the prefix, Serial Dictatorship will choose $x$. Thus all Pareto optimal outcomes have strictly positive probability under RSD. We conclude that $(1)\Leftrightarrow(2)$.

Now we show that $(1)\Leftrightarrow (3)$. Recall that an outcome $x$ is identified with the indicator vector $x\in \{0,1\}^N$ of those agents who deem it acceptable. Then any utilitarian representation of the dichotomous profile of preferences has the following form: $u_i(x)=\lambda_i x_i+c_i$ for some constants $\lambda_i>0$ and $c_i\in \Re$. Indeed, each agent $i$ is indifferent between all her acceptable outcomes, and between all her unacceptable outcomes, so we need to specify only two numbers: $i$'s utility for unacceptable outcomes $c_i$ and for acceptable ones $c_i+\lambda_i$. With this observation, the equivalence $(1)\Leftrightarrow (3)$ becomes a particular case of Criterion~\ref{crit_main}~(d).

The equivalence $(1)\Leftrightarrow (4)$ is a particular case of Criterion~\ref{crit_main}~(f). Indeed, consider the inequality~\eqref{eq_Fishburn_criterion} from the criterion:
$$
\left|\{k\,:\, a_k\succ_i x \}  \right|\geq \left|\{k\,:\, b_k\succ_i x \}  \right|
$$
In the dichotomous domain, if $x$ is acceptable to $i$, then we get zero on both sides and the condition is trivially satisfied. If $x$ is unacceptable to $i$, the inequality boils down to 
$$\sum_{k=1}^K (a_k)_i\geq \sum_{k=1}^K (b_k)_i,$$
i.e., we get inequality~\eqref{eq_Fishburn_dichotomous} from Proposition~\ref{prop:characterizdichotomous}. It remains to check that each the number of times each outcome is repeated in sequences $(a_k)_{k=1,\ldots, K}$ or $(a_k)_{k=1,\ldots, K}$ can be bounded in terms of $|N|$. From Criterion~\ref{crit_main}~(f), we know that the number of repetitions can be bounded by $\exp\left(\frac{1}{2}|X^*|\cdot\ln |X^*|\right)$, where $X^*$ is the set of Pareto optimal outcomes. Recall that the two outcomes are equivalent if all agents are indifferent between them. Clearly, the bound on the number of repetitions remains true if $|X^*|$ is replaced with the number of  non-equivalent outcomes in $X^*$. By Lemma~\ref{lm_Sperner}, the number of such outcomes cannot exceed $\binom{|N|}{\left\lfloor |N|/2 \right\rfloor}$ and thus the number of repetitions is at most $\exp\left(\frac{1}{2}\binom{|N|}{\left\lfloor |N|/2 \right\rfloor}\cdot\ln \binom{|N|}{\left\lfloor |N|/2 \right\rfloor}\right)$.

\subsection{Proof of Proposition~\ref{prop_singlepeak}}\label{sec_singlepeak_proof}
We consider a profile of single-peaked preferences $(\succeq_i)_{i\in N}$ over $X$ and demonstrate that the set of ex-ante efficient lotteries coincides with the set of ex-post efficient lotteries.
It is convenient to enumerate the outcomes in the increasing order according to $\succ$. Hence, without loss of generality, $X=\{1,2,\ldots, |X|\}$ and  $\succ$ coincides with the ordering  $>$ of the natural numbers.  

Recall that $x_i^*$ is the peak of an agent $i$. We denote the leftmost and the rightmost peaks by $l=\min_i x_i^*$ and $r=\max_i x_i^*$, respectively. It is easy to see that the set of Pareto optimal outcomes $X^*$ coincides with the interval $[l,r]=\{x\in X\,:\, l\leq x \leq r\}$. Indeed, all the agents prefer the peak $l$ to any outcome $x< l$ and,  similarly, they prefer $r$ to any $x>r$. Hence, $X^*\subset [l,r]$. For any pair of distinct outcomes  $x, y\in [l,r]$, there is an agent preferring $x$ to $y$ and an agent preferring $y$ to $x$ (e.g., agents with the leftmost and the rightmost peaks), so none of such outcomes can Pareto dominate another. Thus $X^*=[l,r]$.

Let us show that any lottery $p\in \Delta(X^*)$ is ex-ante efficient. By Criterion~\ref{crit_main}~(e), it is enough to demonstrate that that there is no $\alpha\in \Re^{[l,r]}$, except for $\alpha=0$, such that $\sum_{x\in [l,r]} \alpha_x=0$ and $\sum_{y\succeq_i x} \alpha_y\geq 0$ for all $i$ and $x$. Considering agents with the leftmost and the rightmost peaks, we see that $\sum_{y=l}^x \alpha_y\geq 0$ and $\sum_{y=x}^r \alpha_y\geq 0$ for any $x\in [l,r]$. Therefore,
$$0\leq \sum_{y=l}^x \alpha_y=\sum_{y=l}^r \alpha_y - \sum_{y=x+1}^r \alpha_y=0-\sum_{y=x+1}^r \alpha_y\leq 0.$$
Thus $\sum_{y=l}^x \alpha_y=0$ for any $x\in X^*$ and so $\alpha$ itself is equal to zero. We conclude that all $p\in \Delta(X^*)$ are ex-ante efficient.

\subsection{Proof of Proposition~\ref{prop_ballot_and_efficiency} }\label{sec_proof_ballot}
Let us consider a preference profile $(\succeq_i)_{i\in N}$ such that not all ex-post efficient lotteries are ex-ante efficient and show that there is a set $Y\subset X^*$ so that the restriction of the profile to $Y$ is equivalent to a ballot-counting profile. By Criterion~\ref{crit_main}~(f),
there exist two sequences of Pareto optimal outcomes $(a_k)_{k=1,\ldots, K}$ and $(b_k)_{k=1,\ldots, K}$ with no outcomes in common such that
\begin{equation}\label{eq_Fishburn_ballot_proof}
    \left|\{k\,:\, a_k\succ_i x \}  \right|\geq \left|\{k\,:\, b_k\succ_i x \}  \right|
\end{equation}
for any agent $i$ and any outcome $x$.
Let $Y$ be the set of outcomes $y\in X^*$ that enter one of these sequences. We represent each outcome $y\in Y$ by an envelope. If $y$ appears in $(a_k)_{k=1,\ldots, K}$, the envelope contains $\big|\{k\colon a_k=y\}\big|$ slips for candidate $A$ and, if $y$ appears in $(b_k)_{k=1,\ldots, K}$, then the envelope contains $\big|\{k\colon b_k=y\}\big|$ slips for candidate $B$. Then the total number of slips is the same for each of the candidates as the two sequences have the same number of elements. The condition i) of the ballot counting profile is equivalent to~\eqref{eq_Fishburn_ballot_proof}. The condition ii) follows from the fact that all the outcomes in $Y$ are Pareto optimal and, hence, for each pair of distinct $y,z\in Y$, there are agents $i$ and $j$ such that $y\succ_i z$ and $y\prec_j z$. 

It remains to prove  the opposite direction, i.e., that for a profile containing a ballot-counting sub-profile of Pareto optimal outcomes, not all ex-post efficient lotteries are ex-ante efficient. It is enough  to demonstrate that any ballot-counting profile itself has this property. Consider a ballot counting profile with a set of outcomes/envelopes $Z$. Each outcome $z\in Z$ is Pareto optimal by the condition~ii). Let $K=\frac{1}{2}|Z|$ and define two sequences $(a_k)_{k=1,\ldots, K}$ and 
$(b_k)_{k=1,\ldots, K}$. The former is composed of outcomes/envelopes $y$ containing slits for $A$ and the latter is composed of outcomes/envelopes in favor of $B$, and each outcome/envelope is repeated as many times as there are slits in it. Then the condition i) implies~\eqref{eq_Fishburn_ballot_proof}. Moreover, by choosing $x$ to be the second envelope in one of the permutations, we see that~\eqref{eq_Fishburn_ballot_proof} becomes strict. Thus, by Criterion~\ref{crit_main}~(f),  there are ex-post efficient lotteries that are not ex-ante efficient.

\bibliography{main}

\appendix

\section{Linear systems with integral coefficients}\label{app_integral}
Throughout the paper, we rely on linear programming tools.  Here we highlight one of them which is quite basic but apparently not widely known and may be useful beyond applications to ex-ante and ex-post efficiency.
Consider a homogeneous linear system with integral coefficients:
\begin{equation}\label{eq_weak_hom_LP}
A x \leq 0
\end{equation}
where $x\in\Re^m$ is a vector with real coordinates and $A\in \Z^{n\times m}$ is an $n\times m$ matrix composed of integers. This linear system always has a solution $x=0$. We call a solution non-trivial if at least one of the inequalities is strict. Note that a non-zero solution is not necessarily non-trivial.

Fix $k\geq 0$ and consider a $k'\times k'$ sub-matrix $A'$ of $A$ with $0\leq k'\leq k$.
Denote by $q_{k}(A)$ the maximal absolute value of the determinant of $A'$, where the maximum is taken over all such sub-matrices:
$$q_k(A)=\max_{A'\subset A}|\det A'|.$$
For convenience, we agree that the determinant of a $0\times 0$ matrix is equal to $1$. In particular, $q_k(A)\geq 1$ for any $A$ and $k$. 
\begin{lemma}[non-trivial solutions]\label{lm_integrality}
The system~\eqref{eq_weak_hom_LP} has a non-trivial solution if and only if there is a non-trivial \emph{integral} solution $x\in \Z^m$ such that $$|x_j|\leq q_{m-1}(A)$$ for any $j=1,\ldots,m$. 
\end{lemma}
 Matrices with $q_{\min\{n,m\}}(A)=1$ are called totally unimodular. For such matrices, we obtain a well-known result that there is a non-trivial solution $x$ whenever there is such a solution  with $x_j\in\{-1,0,1\}$.
Recall the Hadamard inequality: the absolute value of the determinant of a matrix does not exceed the product of Euclidean norms of its columns. Consequently, $q_k(A)$ admits a simple upper bound
$$q_k(A)\leq  \Big(\sqrt{\min\{n,\,m,\,k\}}\cdot \max_{i,j} |A_{i,j}|\Big)^{\min\{n,\,m,\,k\}}.$$
Indeed, the size of the square submatrix is at most $l=\min\{n,\,m,\, k\}$ and the Euclidean norm of a vector of length $l$ does exceed $\sqrt{l}$ times its maximal element. 
Combining the Hadamard inequality with Lemma~\ref{lm_integrality}, we conclude that for systems with $A_{i,j}\in\{-1,0,1\}$ it is enough to consider integral $x$ with
\begin{equation}\label{eq_bound_x_01_matrix}
|x_j|\leq \exp\left(\frac{\min\{n,m\}}{2}\ln\left( \min\{n,m\}\right)\right).   
\end{equation}

\begin{proof}[Proof of Lemma~\ref{lm_integrality}]
If the problem has a non-trivial integral solution, then, clearly, it has a real one (take the same solution $x$). Hence, one direction  is immediate.
The hard direction is to show that the existence of a non-trivial solution $y\in \Re^m$ implies the existence of a non-trivial integral solution $x\in\Z^m$ satisfying the required bound. 
Let us start from $y$ and construct $x$. 

Without loss of generality, the first inequality is strict for $y$, i.e., $\sum_{j} A_{1,j}\cdot y_j<0$.
Let us add $2m$ extra inequalities 
$$-1 \leq x_j\leq 1, \quad j=1,\ldots,m,$$
to the system  $Ax\leq 0$. The new system has a non-trivial solution $z=\frac{1}{\max_j|y_j|}\cdot y$.  The set of all its solutions is a bounded convex polytope thanks to the added constraints. This polytope contains a point $z$ such that  $\sum_{j} A_{1,j}\cdot z_j<0$. Hence, it has at least one extreme point $z^*$ with $\sum_{j} A_{1,j}z_j^*<0$. As an extreme point of the solution set, it can be obtained as the unique solution to the system of equations formed by active constraints:
$$\left\{\begin{array}{cccc}
    \sum_{j=1}^m A_{i,j}\cdot z_j^*&=&0, & i\in I  \\
     z_j^*&=&1, & j\in J_1\\
     z_j^*&=&-1, & j\in J_{-1}\\
\end{array} \right..$$
for some sets $I\subset\{1,\ldots, n\}$,  $J_1\subset \{1,\ldots,m\}$ and $J_{-1}\subset \{1,\ldots,m\}$. The  values of $z_j^*$ for $j\in J_1\cup J_{-1}$ are given by the second and the third families of equations and we can get rid of these two families by substituting these values in the first family. For the remaining coordinates $(z_j^*)_{j\in \{1,\ldots m\}\setminus (J_1\cup J_{-1})}$, we get a system
$$A' (z_j^*)_{j\in \{1,\ldots, m\}\setminus (J_1\cup J_{-1})}= b',$$
where $A'$ is a square sub-matrix of $A$ corresponding to linearly independent equalities and $b'$ is an integral vector originated as a result of substituting the known coordinates $z_j$ to the first family of equalities. Since $z^*\ne 0$, at least one of the sets $J_1$ and $J_{-1}$ is non-empty. Hence, the dimension of $(z_j^*)_{j\in \{1,\ldots m\}\setminus (J_1\cup J_{-1})}$ is at most $m-1$. Thus $A'$ is a non-degenerate square sub-matrix of $A$ with the dimension of at most $m-1$ and so, its determinant~$D$  enjoys the bound  $1\leq |D|\leq q_{m-1}(A)$.

By Cramer's rule, each coordinate of the solution  $(z_j^*)_{j\in \{1,\ldots m\}\setminus (J_1\cup J_{-1})}$ can be represented as a rational number with the denominator $|D|$. We conclude that $x=|D|\cdot z^*$  is a non-trivial integral solution to the original system. As $z^*$ has all the coordinates between $-1$ and $1$ and $|D|\leq q_{m-1}(A)$, the constructed solution satisfies the desired bound  $|x_j|\leq q_{m-1}(A)$ for all $j$.
\end{proof}

\section{Ex-post Efficiency of Serial Dictatorship}\label{app_SD}

 \begin{proposition}
 For any profile of preferences $(\succeq_i)_{i\in N}$ and any permutation $\sigma$, all outcomes of the Serial Dictatorship  $\mathrm{SD}\big((\succeq_i)_{i\in N},\sigma\big)$ are Pareto optimal.
 \end{proposition}
 \begin{proof}
 Recall that  the lexicographic preference $\succeq_\mathrm{lex}$ associated with a  profile $(\succeq_i)_{i\in N}$ and ordering $\sigma$ is defined as follows: $x\succ_\mathrm{lex} y$ if there exist $k\in \{1,\ldots, n\}$ such that all agents $i=\sigma(1),\ldots, \sigma (k-1)$ are indifferent between $x$ and $y$ and $x\succ_{\sigma(k)} y$. Let $x$ be a maximal element  of $X$ with respect to  $\succeq_\mathrm{lex}$, i.e., $x$ is an outcome of $\mathrm{SD}\big((\succeq_i)_{i\in N},\sigma\big)$. Our goal is to demonstrate that $x$ is Pareto optimal. Towards a contradiction, assume that there is $z\in X$ such that $z\succeq_i x$ for all $i\in N$ and  $z\succ_j x$ for some $j\in N$. Let $k$ be the minimal number such that the agent $i=\sigma(k)$ is not indifferent between $x$ and $z$. By the construction, all agent $i=\sigma(1),\ldots, \sigma (k-1)$ are indifferent between $x$ and $z$ while $i=\sigma(k)$ strictly prefers $z$ to $x$. We conclude that $z \succ_\mathrm{lex} x$ which contradicts the assumption that $x$ is the maximal element and completes the proof.
 \end{proof}

\section{Proof of Theorem~\ref{lm_welfarist}}\label{app_welfare}

In this section we prove Theorem~\ref{lm_welfarist}, i.e., we demonstrate that a lottery $p$ is ex-ante efficient if and only if there is a utilitarian representation $(u_i)_{i\in N}$ of agents' preferences such that $p$ is supported on outcomes $x$ with maximal welfare $\sum_{i\in N} u_i(x)$. 

We will need the Farkas lemma; see \cite{border2013alternative}.
\begin{lemma}\label{lm_Farkas}
A system 
$$\left\{\begin{array}{ccc}
A x&\geq& b\\
C x&=& d
\end{array}
\right.$$ 
with $A\in \Re^{n\times m}$ and $C\in \Re^{k\times m}$ has no solution $x\in \Re^m$ if and only if there exists $\lambda_{\geq}\in \Re_+^n$ and $\lambda_{=}\in \Re^k$ such that $\lambda_{\geq}^T A+ \lambda_{=}^T C =0\in \Re^m $ and $\lambda_\geq\cdot b + \lambda_=\cdot d >0$, where $\cdot$ denotes the dot product. 
\end{lemma}
Recall that the Farkas lemma follows from a general rule of thumb: If a system of linear (in)equalities is infeasible, there is a linear combination of these inequalities  such that the  resulting single inequality is infeasible. The vector $\lambda=(\lambda_\geq,\lambda_=)$ represents the weights in the linear combination. Since we do not impose any restrictions on $x$, the only way to make a single inequality infeasible is to guarantee that the left-hand side \emph{does not depend on $x$.} Hence, we must get zero on the left-hand side and something positive on the right-hand one. 

\begin{proof}[Proof of Theorem~\ref{lm_welfarist}]
Given a preference profile $(\succeq_i)_{i\in N}$ over a set $X$ of outcomes, we need to show that a lottery $p\in \Delta(X)$ is ex-ante efficient if and only if there is a utilitarian representation $(u_i)_{i\in N}$ such that
$p$ is supported on outcomes $x$ with maximal welfare $W(x)=\sum_{i\in N} u_i(x)$.

Necessity is straightforward. Indeed, let us show that if $p$ is supported on welfare-maximizing outcomes for some $(u_i)_{i\in N}$, then $p$ is ex-ante efficient. Towards a contradiction, assume that there is 
$p'\in\Delta(X)$ such that $p'\succeq_i p$ for all $i\in N$ and $p'\succ_j p$ for some $j$. By the weak inequality and the fact that $u_i$ represents $\succeq_i$, we infer that 
$$\sum_{x\in X}p_x'\cdot u_i(x)\geq \sum_{x\in X}p_x\cdot u_i(x)$$
for all $i\in N$ and the inequality is strict for $i=j$.
Summing up these inequalities, we conclude that 
$\sum_{x\in X} p_x'\cdot W(x)> \sum_{x\in X} p_x\cdot W(x)$.
Thus there exists a pair of outcomes $x$ with $p_x>0$ and $y$ with $p_x'>0$ such that $W(y)>W(x)$. This contradicts the assumption that $p$ is supported on outcomes with the maximal welfare.

To prove sufficiency, we start from an ex-ante efficient lottery $p$ and aim to construct a utilitarian representation $(u_i)_{i\in N}$ such that $p$ is supported on welfare-maximizing lotteries. We will show the existence of a utilitarian representation enjoying two additional properties: (1) $u_i(x)-u_i(y)\geq 1$ whenever $i$ strictly prefers $x$ to $y$ (2) all the outcomes $x$ from the support $\supp[p]$ of $p$ have zero welfare (so welfare is non-positive for all outcomes). 

The existence of $(u_i)_{i\in N}$ boils down to feasibility of the following linear system:
\begin{equation}\label{eq_system_u}
\left\{\begin{array}{cccc}
u_i(x)-u_i(y) &\geq& 1, & x \succ_i y\\
u_i(x)-u_i(y) &=& 0, & x \simeq_i y\\
-\sum_{i\in I} u_i(x)&\geq& 0, &  x\in X\setminus\supp[p] \\
\sum_{i\in I} u_i(x)&=& 0, &  x\in \supp[p]
\end{array} \right.,
\end{equation}
where unknowns are $(u_i(x))_{i\in N,\, x\in X}\in \Re^{N\times X}$. Towards a contradiction, assume that this system is infeasible. By the Farkas lemma (Lemma~\ref{lm_Farkas}), the infeasibility of this system means that there is a linear combination of these (in)equalities so that the resulting single inequality is infeasible. 
Denote the corresponding coefficients by $\lambda_{i}(x,y)$ for the first two conditions, $\lambda_-(x)$ for the third condition, and $\lambda_0(x)$ for the last one. 
We obtain the following dual system 
\begin{equation}\label{eq_dual_welfare_proof}
\left\{\begin{array}{cccc}
\sum_{y\colon  x\succeq_i y} \lambda_i(x,y)-\sum_{y\colon y\succeq_i x}\lambda_i(y,x)+\lambda_0(x)&=&0, & x\in \supp[p],\, i\in N\\
\sum_{y\colon  x\succeq_i y} \lambda_i(x,y)-\sum_{y\colon y\succeq_i x}\lambda_i(y,x)-\lambda_-(x)&=&0, & x\in X\setminus\supp[p],\, i\in N
\end{array} \right.
\end{equation}
that must have a solution such that $\lambda_i(x,y)\geq 0$ for $x\succ_i y$, $\lambda_i(x,y)\in\Re$ for $x\simeq_i y$, $\lambda_-(x)\geq 0$ for $x\not\in \supp[p]$, $\lambda_0(x)\in \Re$ for $x\in\supp[p]$, and  $\lambda_j(x,y)>0$ for some $j\in N$ and $x\succ_j y$.
Define $\alpha=\alpha(x)$ by
\begin{equation}\label{eq_alpha_dual_definition}
    \alpha(x)=\sum_{y\colon  x\succeq_i y} \lambda_i(x,y)-\sum_{y\colon y\succeq_i x}\lambda_i(y,x).
\end{equation}
By~\eqref{eq_dual_welfare_proof}, the value of the right-hand side does not depend on $i\in N$ and so the definition of $\alpha$ is correct.

Our goal is to show that a lottery $p'=p+\varepsilon\alpha$ for small positive $\varepsilon$ dominates $p$, i.e., $p'\succeq_i p$ for all agents $i$ and $p'\succ_j p$ for some $j$. 

First of all, we need to check that $p'$ is a lottery, i.e., all the weights are non-negative and sum up to one. From the second equation in~\eqref{eq_dual_welfare_proof}, we obtain 
$$\alpha(x)\geq 0\quad \mbox{for $x\in X\setminus \supp[p]$}$$
and so $p_x'\geq 0$ for all $x\in X$ and small enough $\varepsilon>0$. For each pair of distinct outcomes $z,w\in X$, the coefficient $\lambda_i(z,w)$ enters~\eqref{eq_alpha_dual_definition} with the positive sign for $x=z$ and with the negative sign for $x=w$. Hence,
$$\sum_{x\in X}\alpha(x)=0$$
because of cancellations. We get that $\sum_{x\in X} p_x'=1$, i.e., $p'$ is a lottery. 

Let us check that $p'\succeq_i p$. By the definition of stochastic dominance, this is equivalent to 
\begin{equation}\label{eq_dominance_alpha}
\sum_{x\colon x\succ_i w} \alpha(x)\geq 0
\end{equation}
for all outcomes $w\in X$. By~\eqref{eq_alpha_dual_definition}, we obtain 
$$\sum_{x\colon x\succ_i w} \alpha(x)=\sum_{\scriptsize{\begin{array}{c}
    x\succ_i w\\
     x\succeq_i y
\end{array}}}\lambda_i(x,y)-\sum_{\scriptsize{\begin{array}{c}
     x\succ_i w\\
     y\succeq_i x
\end{array}}}\lambda_i(y,x).$$
The second sum can be rewritten as follows:
$$\sum_{\scriptsize{\begin{array}{c}
     x\succ_i w\\
     y\succeq_i x
\end{array}}}\lambda_i(y,x)=\sum_{\scriptsize{\begin{array}{c}
     y\succ_i w\\
     y\succeq_i x\succ w
\end{array}}}\lambda_i(y,x)=\sum_{\scriptsize{\begin{array}{c}
     x'\succ_i w\\
     x'\succeq_i y'\succ w
\end{array}}}\lambda_i(x',y'),$$
where in the last equality we denoted $x'=y$ and $y'=x$. Thus 
\begin{equation}\label{eq_alpha_uppercontour}
\sum_{x\colon x\succ_i w} \alpha(x)=\sum_{\scriptsize{\begin{array}{c}
     x\succ_i w\\
     x\succeq_i y
\end{array}}}\lambda_i(x,y)-\sum_{\scriptsize{\begin{array}{c}
     x'\succ_i w\\
     x'\succeq_i y'\succ w
\end{array}}}\lambda_i(x',y')=\sum_{\scriptsize{\begin{array}{c}
     x\succ_i w\succeq_i y
\end{array}}}\lambda_i(x,y).
\end{equation}
Thus~\eqref{eq_dominance_alpha} holds thanks to non-negativity of $\lambda_i(x,y)$ for $x\succ_i y$. Moreover, we know that $\lambda_j(x,y)>0$ for agent $j$ and some $x\succ_j y$. Thus, choosing $w=y$ and $i=j$ in~\eqref{eq_alpha_uppercontour}, we see that  $p'\succ_j p$, which contradicts ex-ante efficiency of the lottery $p$. We conclude that the system~\eqref{eq_system_u} has a solution, i.e., there is a utilitarian representation of the preference profile such that $p$ is supported on outcomes with the highest welfare.
\end{proof}

\end{document}